%% file: main.tex
\documentclass[a4paper,UKenglish,cleveref,autoref]{lipics-v2019}
\usepackage{mathtools}
\usepackage{amssymb,amsmath,amsthm}
\usepackage{xspace}
\usepackage[basic,classfont=bold]{complexity}
\usepackage{todonotes}
\usepackage{stmaryrd}
\SetSymbolFont{stmry}{bold}{U}{stmry}{m}{n}
\usepackage{tikz}
\usetikzlibrary{arrows,automata,positioning}
\usepackage{cleveref}

\makeatletter
\let\epsilon\varepsilon
\let\phi\varphi

\let\rho\varrho
\makeatother

\usetikzlibrary{arrows.meta,decorations.markings,circuits,decorations.pathmorphing,
decorations.pathreplacing,arrows,automata,positioning,arrows,calc,fit,
backgrounds,shapes,matrix,intersections} 
\tikzstyle{Player1}=[circle, thick, minimum size=0.6cm, inner 
sep=0cm,draw=black]
\tikzstyle{State}=[circle, font = \small, thick, minimum size=0.9cm, inner 
sep=0cm,draw=black,fill=white]
\tikzstyle{Dummy}=[font = \small]
\tikzstyle{Final}=[circle, accepting, thick, minimum size=0.6cm, inner 
sep=0cm,draw=black]
\tikzstyle{RState}=[circle, very thick, minimum size=0.8cm, inner 
sep=0cm,draw=red]
\tikzstyle{tran}=[draw,->,font=\small,thick]
\tikzstyle{obstyle}=[rounded corners,fill=gray!20]

\renewclass{\P}{P\xspace}
\renewclass{\PSPACE}{PSPACE\xspace}
\renewclass{\EXP}{EXPTIME\xspace}

\def\abs#1{\ensuremath{\lvert #1\rvert}}

\renewcommand{\st}{\mathrel{\mid}}
\newcommand{\distr}{\mathcal{D}}
\newcommand{\prob}{P}
\newcommand{\rew}{R}
\newcommand{\probm}{\mathbb{P}}
\newcommand{\Pv}{\vec{p}}
\newcommand{\Rv}{\vec{v}}
\newcommand{\pathft}[1]{\stackrel{#1}{\rightarrow}}
\DeclareMathOperator*{\argmax}{\mathrm{argmax}}

\newcommand{\calC}{\mathcal{C}}
\newcommand{\states}{S}
\newcommand{\actions}{A}
\newcommand{\transitions}{\prob}

\newcommand{\mdp}{\mathcal{M}}
\newcommand{\Nset}{\mathbb{N}}

\newcommand{\rewards}{R}
\newcommand{\Qset}{\mathbb{Q}}

\newcommand{\odd}{\mathit{odd}}
\newcommand{\even}{\mathit{even}}

\def\code#1{\texttt{#1}}
\def\abs#1{\ensuremath{\lvert #1\rvert}}
\newcommand{\SLP}{S}

\newcommand{\rewval}{\Rv}

\newcommand{\valuation}{\mathit \nu}
\newcommand{\op}{\mathrel{\pm}}

\title{On the Complexity of Value Iteration} 

\author{Nikhil Balaji}{University of Oxford}{nikhil.balaji@cs.ox.ac.uk}{}{}
\author{Stefan Kiefer}{University of Oxford}{stekie@cs.ox.ac.uk}{}{is
supported by a Royal Society University Research Fellowship.}
\author{Petr Novotn\'y}{Masaryk
University}{petr.novotny@fi.muni.cz}{0000-0002-5026-4392}{is supported by the Czech Science Foundation, grant no. GJ19-15134Y ``Verification and Analysis of Probabilistic Programs.''}
\author{Guillermo A. P\'erez}{University of
Antwerp}{guillermoalberto.perez@uantwerpen.be}{0000-0002-1200-4952}{}
\author{Mahsa Shirmohammadi}{CNRS \&
IRIF}{mahsa.shirmohammadi@irif.fr}{}{is
supported by the ``AAPS'' PEPS JCJC grant.}

\authorrunning{N. Balaji et al.}

\Copyright{Nikhil Balaji, Stefan Kiefer, Petr Novotn\'y, Guillermo A. P\'erez,
and Mahsa Shirmohammadi}



\acknowledgements{We thank James Worrell for
helpful comments on early version of this work.}

\nolinenumbers
\hideLIPIcs

\EventEditors{Christel Baier, Ioannis Chatzigiannakis, Paola Flocchini, and Stefano Leonardi}
\EventNoEds{4}
\EventLongTitle{46th International Colloquium on Automata, Languages, and Programming (ICALP 2019)}
\EventShortTitle{ICALP 2019}
\EventAcronym{ICALP}
\EventYear{2019}
\EventDate{July 9--12, 2019}
\EventLocation{Patras, Greece}
\EventLogo{eatcs}
\SeriesVolume{132}
\ArticleNo{97}

\begin{document}

\maketitle

\begin{abstract}
  Value iteration is a fundamental algorithm for solving Markov Decision
  Processes (MDPs).  It computes the maximal $n$-step payoff by iterating $n$
  times a recurrence equation which is naturally associated to the MDP.  At
  the same time, value iteration provides a policy for the MDP that is optimal
  on a given finite horizon $n$.  In this paper, we settle the computational
  complexity of value iteration.  We show that, given a horizon $n$ in binary
  and an MDP, computing an optimal policy is \EXP-complete, thus resolving an
  open problem that goes back to the seminal 1987 paper on the complexity of
  MDPs by Papadimitriou and Tsitsiklis. To obtain this main result, we develop
  several stepping stones that yield results of an independent interest. For
  instance, we show that it is \EXP-complete to compute the $n$-fold iteration
  (with $n$ in binary) of a function given by a straight-line program over the
  integers with $\max$ and $+$ as operators. We also provide new complexity
  results for the bounded halting problem in linear-update counter machines.
\end{abstract}

\section{Introduction}
\input{intro}

\section{Markov Decision Processes and Finite-Horizon Problems}
\input{new-mdps}

\section{Straight-Line Programs and The Powering Problem}
\input{new-slps}

\section{Main Reductions}\label{sec:reductions}
\input{reduction}

\clearpage
\bibliography{refs}

\clearpage
\appendix
\begin{center}
{\LARGE Technical Appendix}
\end{center}
\input{appendix}

\end{document}

%% file: intro.tex
Markov decision processes (MDP) are a fundamental formalism of decision making
under probabilistic uncertainty~\cite{puterman05,bertsekas2005dynamic}. As
such, they play a prominent role in numerous domains, including artificial
intelligence and machine
learning~\cite{sutton2018reinforcement,sigaud2013markov}, control
theory~\cite{blondel2000survey,NIPS2004_2569}, operations research and
finance~\cite{bauerle2011finance,schal2002markov}, as well as formal
verification~\cite{ModCheckHB18,ModCheckPrinciples08}, to name a few.
Informally, an MDP represents a system which is, at every time step, in one of
the \emph{states} from a finite set $S$. The system evolves in steps: in each
step, we can perform an \emph{action} (or \emph{decision}) from a finite set
$A$. When using an action $a\in A$ in state $s\in S$, we collect an immediate
\emph{reward} $\rewards(s,a)$ and then transition stochastically to a
successor state according to a rational-valued distribution
$\transitions(s,a)$, which is given as a part of the MDP. This interaction
with an MDP proceeds over either a \emph{finite} or \emph{infinite} horizon.
In the finite-horizon case, we are given a bound $H\in \Nset$ (a
\emph{horizon}) such that the interaction stops after $H$ steps; in the
infinite horizon case the process goes on forever. To \emph{solve} an MDP
means to find an optimal \emph{policy}; that is, a blueprint for selecting
actions that maximizes the expected reward accumulated over a finite or
infinite horizon. The accumulated rewards are typically \emph{discounted} by
some factor $0<\gamma \leq 1$; for infinite horizon, we need $\gamma < 1$ to
ensure that the infinite sum is well defined.

\subparagraph{Value iteration.} Given the importance of MDPs, it is hardly
surprising that they have attracted significant interest in the theory
community. Past research on MDPs included the study of complexity
issues~\cite{pt87} as well as the design and analysis of algorithms for
solving
MDPs~\cite{DBLP:conf/uai/LittmanDK95,DBLP:conf/uai/MansourS99,ye2005new,ye2011simplex}.
In this paper, we provide a fresh look on one of the most familiar algorithms
for MDPs: \emph{value iteration (VI)}. Introduced by Bellman in the
1950s~\cite{bellman2013dynamic}, VI makes use of the optimality principle: the
maximal $n$-step reward achievable from a state $s$, which we denote by
$\rewval_n(s)$, satisfies the recurrence
\begin{equation}\label{eq:intro-bellman}
    \rewval_n(s) = \max_{a\in \actions}\left\{ \rewards(s,a) + \gamma \cdot
      \sum_{s' \in
      \states} \transitions(s,a)(s')\cdot \rewval_{n-1}(s')
    \right\}\enspace,
\end{equation} 
with $\rewval_0(s)=0$.  Consequently, a finite-horizon policy is optimal if
and only if it chooses, in a situation when the current state is~$s$ and $n$
steps are remaining, an action maximizing the right-hand side (RHS)
of~\eqref{eq:intro-bellman}. Thus, to solve an MDP with a finite horizon $H$,
the VI algorithm computes the values $\rewval_n(s)$ for all $0\leq n \leq H$
and all states $s$, by iterating the recurrence~\eqref{eq:intro-bellman}.
Using these values, VI then outputs (using some tie-breaking rule) some policy
satisfying the aforementioned optimality characterization. VI can be deployed
also for infinite-horizon MDPs: one can effectively compute
a horizon $H$ such that action $a$ is optimal in state $s$ for an infinite
horizon\footnote{In infinite-horizon MDPs, there is always an optimal
\emph{stationary} policy, which makes decisions based only on the current
state.~\cite{puterman05}} if it maximizes the RHS of~\eqref{eq:intro-bellman}
for $n=H$~\cite{Bertsekas:1987:DPD:26970}. This $H$ has a bit-size which is
polynomial in the size of the original MDP, but the magnitude of $H$ can be
exponential in the size of the MDP if the discount factor is given in
binary~\cite{DBLP:conf/uai/LittmanDK95}. 

VI is one of the most popular MDP-solving algorithms due to its versatility
(as shown above, it can be used for several MDP-related problems) and
conceptual simplicity, which makes it easy to implement within different
programming paradigms~\cite{vicav18,DBLP:conf/ecai/WuHGZL16}, including
implementation via neural nets~\cite{VINetworks16}. Several variants of VI
with improved performance were
developed~\cite{tseng,DBLP:journals/jair/DaiMWG11}. For instance, the recent
paper by Sidford et al.~\cite{sidford18soda} presented a new class of
randomized VI techniques with the best theoretical runtime bounds (for certain
values of parameters) among all known MDP solvers. The paper also expresses
hope that their techniques ``\emph{will be useful in the development of even
faster MDP algorithms.}'' To get insight into the underlying structure of VI,
which might enable or limit further such accelerations, we take a
complexity-theoretic vantage point and study the theoretical complexity of
computing an \emph{outcome} of a VI execution. That is, we consider the
following decision problem \textsc{ValIt}: given an MDP with a finite horizon
$H$ (encoded as a binary number), does a given action $a$ maximize the RHS
of~\eqref{eq:intro-bellman} for $n=H$? This problem is inspired by the paper
of Fearnley and Savani~\cite{Fearnley15stoc}, where they show \PSPACE-hardness
(and thus also completeness) for the problem of determining an outcome of
\emph{policy iteration}, another well-known algorithm for MDP solving. To the
best of our knowledge, VI has not yet been explicitly subjected to this type
of analysis. However, questions about the complexity of \textsc{ValIt} were
implicitly raised by previous work on the complexity of finite-horizon MDPs,
as discussed in the next paragraph.

\subparagraph{Finite-horizon MDPs.} The complexity of finite-horizon MDPs is a
long-standing open problem. Since ``finding an optimal policy'' is a function
problem, we can instead consider the decision variant: ``In a given
finite-horizon MDP, is it optimal to use a given action in the first step?'' As
discussed above, this is exactly the \textsc{ValIt} problem in disguise.

In the seminal 1987 paper on the complexity of MDPs~\cite{pt87}, Papadimitriou
and Tsitsiklis showed \P-completeness of a special case of finite-horizon
optimization where the horizon $H$ has magnitude polynomial in the size of the
MDP. At the same time, they noted that in the general case of binary-encoded
$H$, VI can be executed on an \EXP-bounded Turing machine (since $H$ is
represented using $\log(H)$ bits, the number of iterations is exponential in
the size of the input). Hence \textsc{ValIt} is in \EXP. However, the exact
complexity of the general finite-horizon optimization remained open ever
since, with the best lower bound being the \P-hardness inherited from the
``polynomial $H$'' sub-problem.
Tseng~\cite{tseng} presented a more efficient (though
still exponential) algorithm for finite-horizon MDPs satisfying a
certain stability condition; in the same paper, he comments that \emph{``in view
of the stability assumptions needed to obtain an exact solution and the absence
of negative results, we are still far from a complete complexity theory for this
problem.''}

In this paper, we address this issue, provide the missing negative results,
and \emph{provide tight bounds on the computational complexity} of
\textsc{ValIt} and finite-horizon MDP optimization. 

\subsection*{Our Results} The main result of the paper is that \textsc{ValIt}
is \EXP-complete (\Cref{th-finite-rew}). In the rest of this section, we first
explain some challenges we needed to overcome to obtain the result. Then we
sketch our main techniques and conclude with discussing the significance of
our results, which extends beyond MDPs to several areas of independent
interest.

\subsection*{Challenges} \subparagraph*{Bitsize of numbers.} One might be
tempted to believe that \textsc{ValIt} is in \PSPACE, since the algorithm
needs to store only polynomially many values at a time. However, the bitsize
of these values may become exponentially large during the computation (e.g.,
the quantity $\rewval_n(s)$ may halve in every step). Hence, the algorithm
\emph{cannot} be directly implemented by a polynomial-space Turing machine
(TM). One could try to adapt the method of Allender et al.~\cite{hab02,abd14}
based on an intricate use of the Chinese remainder representation (CRR) of
integers. However, there is no known way of computing the $\max$ operation
directly and efficiently on numbers in CRR.

\subparagraph*{Complex optimal policies.} Another hope for \PSPACE{}
membership would be a possibly special structure of optimal policies. Fixing
any concrete policy turns an MDP into a Markov chain, whose $H$-step behavior
can be evaluated in polynomial space (using, e.g., the aforementioned CRR
technique of Allender et al.). If we could prove that (A) an optimal policy
can be represented in polynomial space and (B) that the Markov chain induced
by such a policy is polynomially large in the size of the MDP, we would get
the following \PSPACE{} algorithm: cycle through all policies that satisfy (A)
and (B), evaluate each of them, and keep track of the best one found so far. 
Tseng~\cite{tseng} commented that optimal policies in finite-horizon MDPs are
``poorly understood''. Hence, there was still hope that optimal Markovian
deterministic policies may have a shape that satisfies both (A) and (B).
Unless \PSPACE{} = \EXP{}, our results put such hopes to rest.

\subparagraph*{No hardness by succinctness.} One might try to prove
\EXP-hardness using a \emph{succinctness} argument. The results of~\cite{pt87}
show that \textsc{ValIt} is $\P$-hard when the horizon is written in unary,
and many optimization problems over discrete structures incur an exponential
blow-up in complexity when the discrete structure is encoded succinctly, e.g.,
by a circuit~\cite{py86}. Giving a horizon $H$ in binary amounts to a succinct
encoding of an exponentially large MDP obtained by ``unfolding'' the original
MDP into a DAG-like MDP of depth $H$. This unfolded MDP is ``narrow''
in the sense that it consists of many polynomial-sized layers, while standard
$\EXP$-hardness-by-succinctness proofs, use succinct structures of an
exponential ``width'' and ``depth'', accommodating the tape contents of an
\EXP-bounded TM. Hence, straightforward succinctness proofs do not apply here;
e.g., there does not seem to be a direct reduction from the succinct circuit
value problem.

\subsection*{Our Techniques} To obtain \EXP-hardness of \textsc{ValIt}, we
proceed by a sequence of non-trivial reductions. Below we outline these
reductions in the order in which they appear in the sequence, see
Figure~\ref{fig:chain}. In the main text, we present the reductions in a
different order (indicated by the numbering of propositions and theorems), so
that we start with MDPs and gradually introduce more technical notions. 

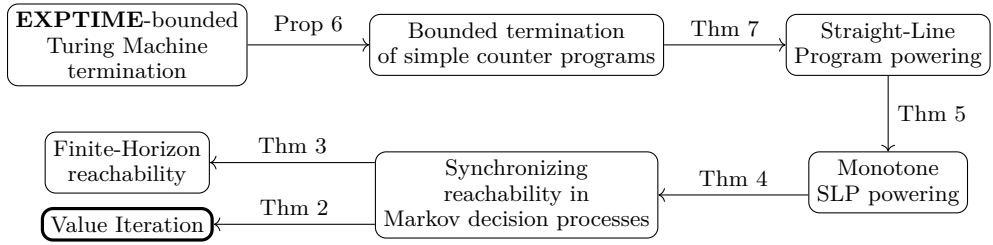
\begin{figure}
\centering
\tikzstyle{prob}=[rounded corners, inner sep=3pt,rectangle, text centered,
align=center, draw, font=\footnotesize]
\tikzstyle{reduct}=[->,font=\footnotesize]
\begin{tikzpicture}
\node[prob] (TM) {\EXP-bounded\\ Turing Machine\\ termination};
\node[prob,right=1.6cm of TM] (CP) {Bounded termination\\ of simple counter programs};
\node[prob,right=1.6cm of CP] (POW) {Straight-Line\\ Program powering};
\node[prob,below=1cm of POW] (MonPOW) {Monotone\\ SLP powering};
\node[prob,below=1cm of CP] (Sync) {Synchronizing\\
reachability in\\ Markov decision processes};
\node[prob,below=0.6cm of TM] (Reach) {Finite-Horizon\\ reachability};
\node[prob,below=1.6cm of TM,very thick] (Val) {Value Iteration};
\draw[reduct] (TM) -- node[above] {Prop~\ref{pro:exp-hardness-simple-cms}}  (CP);
\draw[reduct] (CP) -- node[above] {Thm~\ref{thm:counter-to-slppow}}  (POW);
\draw[reduct] (POW) -- node[right] {Thm~\ref{thm:monotone-hard}}  (MonPOW);
\draw[reduct] (MonPOW.west|-Sync) -- node[above]
  {Thm~\ref{thm:power-to-sync}}  (Sync);
\draw[reduct] (Sync.west|-Val) -- node[above] {Thm~\ref{th-sync-to-rew}}  (Val);
\draw[reduct] (Sync.west|-Reach) -- node[above] {Thm~\ref{th-finite-reach}}  (Reach);
\end{tikzpicture}
\caption{Chain of reductions.}
\label{fig:chain}
\end{figure}

We start from a canonical \EXP-complete problem: the halting problem for an
exponential-time TM.  We then present a reduction to a halting problem for a
class of counter programs (CPs; simple imperative programs with integer
variables) that allow for linear variable updates. In this way, we encode the
tape contents into numerical values (\ref{pro:exp-hardness-simple-cms}). The
crucial feature of this reduction is that the produced CP possesses a special
\emph{simplicity} property, which imposes certain restrictions on the use of
tests during the computation. 

Next, we introduce \emph{straight-line programs (SLPs)} with $\max$, $+$, and
$-$ operations. SLPs are a standard model of arithmetic
computation~\cite{allender-numerical} and they can be equivalently viewed as
arithmetic circuits consisting (in our case) of $\max$, $+$, and $-$ gates. We
also consider a sub-class of SLPs with only $\max,+$ operations, so called
\emph{monotone SLPs.} We define the following \emph{powering problem:} given a
function $f\colon \Qset^n \rightarrow \Qset^n$ represented as an SLP, a
horizon $H$, an initial argument $\vec{x}\in \{0,1\}^n$, and two indices
$1\leq i,j\leq n$, is it true that the $i$-component of $f^{H}(\vec{x})$, i.e.
the image of $\vec{x}$ with respect to the $H$-fold composition of $f$, is
greater than the $j$-component of $f^{H}(\vec{x})$? Although VI in MDPs does
not necessarily involve integers, the powering problem for monotone SLPs
captures the complexity inherent in iterating the
recurrence~\eqref{eq:intro-bellman}.  To obtain a reduction from CPs to SLP
powering, we construct SLP gadgets with $\max$, $+$ and $-$ (minus) operations
to simulate the tests in CPs; the simplicity of the input CP is crucial for
this reduction to work (\Cref{thm:counter-to-slppow}). To get rid of the minus
operation, we adapt a technique by Allender et al.~\cite{AKM}, which
introduces a new ``offset'' counter
and models
subtraction by increasing the value of the offset (\Cref{thm:monotone-hard}).

A final step is to show a reduction from monotone SLP powering to
\textsc{ValIt}. The reduction proceeds via an intermediate problem of
\emph{synchronizing reachability in MDPs} (maximize the probability of being
in a target set of states $T$ \emph{after exactly} $H$ steps \cite{dms14}).
This divides a rather technical reduction into more comprehensible parts. We
present novel reductions from monotone SLP powering to synchronizing
reachability~(\Cref{thm:power-to-sync}), and from the latter problem to
\textsc{ValIt} (\Cref{th-sync-to-rew}). As a by-product, we present a
reduction proving \EXP-hardness of finite-horizon reachability in MDPs,
arguably the conceptually simplest objective in probabilistic
decision-making~(\Cref{th-finite-reach}).

\subsection*{Significance}
As our main result, we characterize the complexity of computing an outcome of
VI, one of the fundamental algorithms for solving both finite- and
infinite-horizon MDPs. As a consequence, we resolve a long-standing complexity
issue~\cite{pt87} of solving finite-horizon MDPs.

On our way to proving this result, we encounter non-trivial stepping stones
which are of an independent interest.  First, we shed light on the complexity
of succinctly represented arithmetic circuits, showing that comparing two
output wires of a given $(\max,+)$-circuit incurs an exponential blow-up in
complexity already when employing a very rudimental form of succinctness:
composing a single $(\max,+)$-circuit with $H$ copies of itself, yielding a
circuit of exponential ``height'' but only polynomial ``width.''
Second, we obtain new hardness results for the bounded reachability  
problem in linear-update counter programs. CPs are related to  
several classical abstractions of computational machines, such as Minsky  
machines and Petri nets~\cite{Petri-STOC}, see~\cite{Petri-elementary} for  
a recent breakthrough in this area. Our work establishes a novel  
connection between counter programs and MDPs.

\subsection*{Further Related Work} Our work is also related to a series of
papers on finite-horizon
planning~\cite{littman1997probplan,goldsmith1997complexity,GMcomplexity,littman1998jair}.
The survey paper \cite{Mundhenk-fh-survey} provides a comprehensive overview
of these results. These papers consider either MDPs with a polynomially large
horizon, or \emph{succinctly} represented MDPs of possibly exponential
``width'' (the succinctness was achieved by circuit-encoding). The
aforementioned hardness-by-succinctness proofs are often used here. The
arbitrary horizon problem for standard MDPs, which we study, is left open in
these papers, and our work employs substantially different techniques. The
complexity of finite-horizon \emph{decentralized} MDPs was studied
in~\cite{dec-mdp}.

%% file: new-mdps.tex
We start with some preliminaries. A \emph{probability distribution}~$d:S\to [0,1]$ over a finite set~$\states$ is  a 
function such that $\sum_{s\in \states} d(s)=1$.
We denote by $\distr(\states)$ the set of all (rational) probability distributions over~$S$. 
The \emph{Dirac} distribution on~$s\in \states$ 
assigns probability 1 to~$s$.

A Markov decision process (MDP)~$\mdp=(\states,\actions,\prob, \rew,\gamma)$ consists of
a finite set~$\states$ of \emph{states}, a finite set~$\actions$  of \emph{actions},
a  \emph{transition  function} $\prob\colon \states \times\actions \to \distr(\states)$,
a \emph{reward function} $\rew: \states \times \actions \to \mathbb{Q}$, and 
a \emph{discount factor}~$\gamma\in (0,1]$.
The transition function $\prob$ assigns to each state~$s$ and action~$a$ a distribution over 
the successor states, while the reward function assigns to $s$ and $a$ a rational reward.

A \emph{path}~$\rho$ is an alternating sequence $s_0 a_1 s_1  \cdots a_n s_n$ of 
visited states and played actions in~$\mdp$ (that starts and ends in a state);
 write $\abs{\rho}=n$ for the length of~$\rho$. We may use $s_0 \pathft{\rho} s_n$	
to denote that path~$\rho$ goes from $s_0$ to $s_n$.
We extend the reward function $\rew$ from  single  state-action pairs to 
paths by $\rew(\rho)=\sum_{1\leq i\leq n} \rew(s_{i-1},a_i) \gamma^{i-1}$.

A \emph{policy} for the controller is  a function~$\sigma$ that assigns to 
each path a distribution over actions.
Let $\probm_{\mdp,s,\sigma}(\rho)$ denote
the probability of a path $\rho$ starting in~$s$ when the controller 
follows the policy~$\sigma$.
This probability is  defined inductively by setting
$\probm_{\mdp,s,\sigma}(s_0)=1$ if $s=s_0$, and $\probm_{\mdp,s,\sigma}(s_0)=0$ otherwise. 
For a path $\rho=s_0 a_1 s_1 \cdots s_{n-1}  a_n s_n$, we set
\[
  \probm_{\mdp,s,\sigma}(\rho)=\probm_{\mdp,s,\sigma}(s_0 \cdots s_{n-1})
  \cdot \sigma(s_0 \cdots s_{n-1})(a_n) \cdot \prob(s_{n-1},a_n)(s_n)\enspace.
\]
We omit the subscripts  from $\probm_{\mdp,s,\sigma}(\cdot)$ if they are clear from the context. Additionally, we extend $\probm_{\mdp,s,\sigma}(\cdot)$ to sets of paths of the same length by summing the probabilities of all the paths in the set.

In this paper, we focus on a special class of policies:
A \emph{(deterministic) Markov} policy 
	is a function $\sigma: \mathbb{N}\times \states \to \actions$.
Intuitively, a controller following a Markov policy plays $\sigma(n,s)$ from $s$ if it is the $n$-th visited state, irrespective of the other states in the path.
	Markov policies suffice for the problems we consider.

%
%

\subsection{Finite-Horizon  Problems}\label{sec:fin-horizon-probs}
Given an MDP~$\mdp$, the core problem of MDPs is computing the values of states 
with respect to 
the \emph{maximum expected  reward}.
Let $\Rv_n \in \mathbb{Q}^{\states}$ denote the vector of $n$-step maximum expected rewards obtainable from each state of the MDP. That is, for all $s \in S$ we have that
\[ 
  \Rv_n (s) = \max_{\sigma} \left(\sum_{\abs{\rho}=n}
  \probm_{s,\sigma}(\rho) \cdot \rew(\rho) \right)\enspace.
\]
Note that $\Rv_0=\vec{0}$ by this definition. The vector~$\Rv_{n}$ can be
computed by \emph{value iteration}, i.e. by iterating the recurrence stated in
Equation~\eqref{eq:intro-bellman}.  From that recurrence, for each $n\in
\mathbb{N}$ and state~$s_0$, one can extract an (optimal) Markov policy
$\sigma$ that achieves the maximum value $\Rv_{n}(s_0)$ after $n$ steps: for
each $s\in \states$ and for $1\leq i\leq n$ we have
\[ 
 \sigma(i-1,s)=\argmax_{a \in \actions} \left\{  \rew(s,a) + \gamma \cdot
   \sum_{s' \in
 \states} \prob(s,a)(s') \cdot \Rv_{n-i}(s') \right\}\enspace.
\] 

\begin{figure}[t]
\begin{center}
\scalebox{.8}{
\begin{tikzpicture}
\node[draw=none] at (0, 1.7)   (n) {MDP $\mathcal{N}$};
\node[state] at (0, 0)   (s1) {$s$};
\node[state] at (0, -2)   (s2) {$s_1$};
 \node[state] at (2.5, 1.5)   (s3) {$s_2$};
\node[state] at (5, 0)   (s4) {$t$};
 \node[state] at (2.5, 0)   (s5) {$s_3$};

\path[->] (s1)  edge [bend right] node[midway, left, align=center] {$a,0:\frac{1}{2}$\\$b,0:1$}  (s2) ;
\path[->] (s2) edge [bend right]  node[midway, right] {$a,b,2:1$} (s1) ;

\path[->] (s1) edge  node[near end, left] {$a,0:\frac{1}{2}~~$} (s3) ;
\path[->] (s3) edge  node[near end, above] {$~~~~~~a,b,2:1$} (s4) ;
\path[->] (s4) edge  node[midway, below] {$a,b,1:1$} (s5) ;
\path[->] (s5) edge  node[midway, below] {$a,b,0:1$} (s1) ;

\node[draw=none] at (11, 0)   (s1) {
  $\begin{aligned}
    \Rv_{n}(s)=&\max\left(\frac{1}{4} \Rv_{n-1}(s_1)+\frac{1}{4}\Rv_{n-1}(s_2),
      \frac{1}{2}\Rv_{n-1}(s_1)\right)\\
    \Rv_{n}(s_1)=&2+\frac{1}{2}\Rv_{n-1}(s)\\
    \Rv_{n}(s_2)=&2+\frac{1}{2}\Rv_{n-1}(t)\\
    \Rv_{n}(t)=&1+\frac{1}{2}\Rv_{n-1}(s_3)\\
    \Rv_{n}(s_3)=&\frac{1}{2}\Rv_{n-1}(s)
  \end{aligned}$};

%

\end{tikzpicture}
}
\caption{ The transitions are labelled with actions, rewards and their probabilities. For example, 
the reward of the transition from~$s$ to $s_1$ on action~$a$ is $0$, and its probability is~$\frac{1}{2}$.
}
\label{fig:valreward}
\end{center}
\vspace{-.7cm}
\end{figure}

Papadimitriou and Tsitsiklis posed the {\bf finite-horizon reward
problem} which asks to compute such an optimal policy for the controller~\cite{pt87}. 
Formally, given an MDP $\mdp$, an initial state $s_0 \in S$, a distinguished action $a \in A$, and a horizon $H \in \mathbb{N}$ encoded in binary,
the finite-horizon reward problem asks whether there exists a 
policy achieving $\Rv_H(s_0)$
by choosing~$a$ as the first action from~$s_0$. 
Note that this problem is equivalent to the \textsc{ValIt} problem defined in the introduction.

Consider the MDP~$\mathcal{N}$ depicted in Figure~\ref{fig:valreward} with~$\gamma=\frac{1}{2}$.
By iterating the indicated recurrence, 
we have that $\Rv_5(s)=\max(\frac{1}{4} \Rv_4(s_1)+\frac{1}{4}\Rv_4(s_2)  , \frac{1}{2}\Rv_4(s_1))=\frac{41}{32}$.
The value of $\Rv_5(s)$ is due to the second
argument of $\max$ (corresponding to  action~$b$), hence
a policy to maximize~$\Rv_5(s)$ starts with~$b$ in~$s$. 

The finite-horizon reward problem can be
decided by value iteration in exponential time by unfolding
recurrence~\eqref{eq:intro-bellman} for~$H$ steps~\cite{puterman05}, while 
the best known lower bound is \P-hardness~\cite{pt87}.
Our main result closes this long-standing complexity gap:

\begin{theorem}\label{th-finite-rew}
The finite-horizon reward problem (and thus also the \textsc{ValIt} problem) is \EXP-complete.
\end{theorem}


To prove \EXP-completeness of the finite-horizon reward problem, we introduce 
a variant of  reachability, which we  call \emph{synchronized reachability}~\cite{dms14}.
Let $t\in \states$ be a target state. 
For  reachability, the objective is to maximize the probability of taking a path from $s$ to $t$, whereas in 
synchronized reachability only a subset of such paths with the same length are considered.

%
Let $\mdp$ be an MDP,~$s_0$ an initial state, and $a$ an action.
Define $\Pv_{\leq n} \in \mathbb{Q}^{\states}$ as the vector of maximum probabilities of taking a path to $t$ within $n$ steps.
Similarly, define $\Pv_{=n} \in \mathbb{Q}^{\states}$ to be
the vector of maximum probabilities
of taking such a path with length exactly~$n$. Formally, for all $s \in S$ we have that
%
\[
  \Pv_{\leq n}(s)=\max_{\sigma} \left( \probm_{s,\sigma} (\{s \pathft{\rho} t :
    \abs{\rho}\leq n\}) \right) \text{ and }
  \Pv_{= n}(s) = \max_{\sigma} \left( \probm_{s,\sigma} (\{s \pathft{\rho} t
  : \abs{\rho}=n\}) \right)\enspace.
\]
%
Given a horizon $H$, encoded in binary,
the {\bf finite-horizon reachability problem} asks 
	whether an optimal policy achieving $\Pv_{\leq H}(s_0)$
 chooses action~$a$ as the first action from~$s_0$; 
the {\bf finite-horizon synchronized-reachability problem} asks
	whether an optimal policy achieving $\Pv_{= H}(s_0)$
chooses action~$a$ as the first action from~$s_0$.

\subsection{Connections Among Finite-Horizon  Problems}
We now prove the following theorem.

\begin{theorem}\label{th-sync-to-rew}
The finite-horizon synchronized-reachability problem reduces, in polynomial time, to the finite-horizon reward problem.
\end{theorem} 

Consider an MDP~$\mdp$, an initial state~$s_0$, an action~$a$ and a target state~$t$. The following recurrence can be used to compute
$\Pv_{= n}(s)$:
\begin{equation}\label{eq-rec-synch}
  \Pv_{=n}(s)=\max_{a \in \actions} \left\{  \sum_{s'\in \states}  \prob(s,a)(s') \cdot
  \Pv_{=n-1}(s') \right\}\enspace,
\end{equation}
where $\Pv_0(t)=1$ and $\Pv_0(s)=0$ for all $s \neq t$.
We construct a new MDP~$\mathcal{N}$
obtained from~$\mdp$ by replacing 
 all transitions by two consecutive transitions.
The construction is such that the
probability of going from~$s$ to $t$ with a path of length~$n$ in~$\mdp$
is equal to the probability of going from~$s$ to $t$ with a path of  length~$2n$ in~$\mathcal{N}$. 
More formally, for all $s,s'$ and $a$ with $\prob(s,a)(s')=p$, the
transition
\scalebox{.8}{
\tikz[baseline=-0.5ex,
state/.style={inner sep=0cm,circle,draw,minimum size=0.6cm}]{
\node[state] at (0, 0)   (a) {$s$};
\node[state] at (2, 0)   (c) {$s'$};

\path[->] (a)  edge  node[midway, above]  {$-:p$} (c) ;
}}  is replaced 
 with 
\scalebox{.8}{
\tikz[baseline=-0.5ex,
state/.style={inner sep=0cm,circle,draw,minimum size=0.6cm}]{
\node[state] at (0, 0)   (a) {$s$};
 \node[state] at (2, 0)   (b) {};
\node[state] at (4, 0)   (c) {$s'$};

\path[->] (a)  edge node[midway, above] {$1:p$} (b) ;
\path[->] (b) edge  node[midway, above] {$0:1$} (c) ;
}}
if $s = t$ and with
\scalebox{.8}{
\tikz[baseline=-0.5ex,
state/.style={inner sep=0cm,circle,draw,minimum size=0.6cm}]{
\node[state] at (0, 0)   (a) {$s$};
 \node[state] at (2, 0)   (b) {};
\node[state] at (4, 0)   (c) {$s'$};

\path[->] (a)  edge node[midway, above] {$0:p$} (b) ;
\path[->] (b) edge  node[midway, above] {$\frac{1}{\gamma}:1$} (c) ;
}}
otherwise; where~$0<\gamma\leq 1$ is an arbitrary chosen 
discount factor for $\mathcal{N}$, and 
 the intermediate state in both cases
is a new state.
The MDP $\mathcal{N}$ in Figure~\ref{fig:valreward} is the result of applying the construction to $\mdp$
in Figure~\ref{fig:connection} 
with $\gamma=\frac{1}{2}$.

For the constructed MDP $\mathcal{N}$, one can show that for all states~$s$,
an action is optimal to maximize $\Pv_{=2H}(s)$ if and only if it is optimal to maximize $\Rv_{2H+1}(s)$.
Consider the MDPs from Figure~\ref{fig:valreward} as an example.
We have previously argued that
a policy maximizing $\Rv_{5}(s)$ in $\mathcal{N}$ starts with action~$b$.
Observe that the optimal first choice to maximize $\Pv_{4}(s)$ is also~$b$.
This implies that an optimal policy of~$\mdp$ for synchronized-reachability with $H=2$
starts with~$b$, too.
By the above argument,
the finite-horizon synchronized-reachability problem reduces to the finite-horizon reward problem.

Hence, to obtain \Cref{th-finite-rew}, it remains to determine the complexity of the finite-horizon synchronized-reachability
problem.
To this aim, we show a close connection between MDPs and a class of piecewise-affine functions
represented by \emph{straight line programs} (SLPs). Section~\ref{sec-slps} provides the details.

\subparagraph{Finite-horizon reachability.}
We also show the finite-horizon synchronized-reachability problem 
reduces to the finite-horizon reachability problem.
We remark that the natural probability-$1$ variants
of these  problems have different complexities:
specifically, the problem
of reaching $t$ from $s$ within $H$ steps with probability~$1$ 
is  in \P;
however, the analogous problem of
reaching $t$ from $s$ in exactly $H$ steps with probability~$1$ 
is \PSPACE-complete~\cite{dms14}.

\begin{theorem}\label{th-finite-reach}
The finite-horizon synchronized reachability problem reduces, in polynomial time, to the finite-horizon reachability problem. 
\end{theorem}



%% file: new-slps.tex
\label{sec-slps}

We now establish the connection between MDPs and SLP powering. We start with
preliminaries.

For all $n\in \mathbb{N}$, define the set~$var_n:=\{x_1,\ldots,x_n\}$ of
variables and the collection of terms 
\[
  \mathcal{T}_n:=\{ a_1x_{j_1}+\cdots+a_nx_{j_n}+b \mid a_i,b\in \{-1,0,1\}
  \text{ and } 1 \leq j_i \leq n, \text{ for all } 1\leq i\leq n\}.
\]
A \emph{straight-line program} (SLP) of \emph{order}~$n$ is a sequence $c_1,\dots,c_m$ of \emph{commands}
of the form $x \leftarrow \max(T)$, where $x\in var_n$ and $T\subseteq \mathcal{T}_n$ is non-empty.
We refer to commands $x \leftarrow b$ as \emph{initializations}. 
Recall that $\min(x,y) = -\max(-x, -y)$.

For complexity analyses we shall assume that $T$, for every command, is given
explicitly as a list of terms. Each term is also assumed to be explicitly
represented as a constant, a list of coefficients $a_i$, and a list of indices
$j_i$, both lists having length $n$ (i.e. the number of variables). The size
of $T$, and also that of the command, corresponds to the length of its list of
terms; the size of the SLP, the sum of the sizes of its commands.

A \emph{valuation} $\valuation$ is a vector in $\mathbb{Z}^{n}$, 
where the $i$-th coordinate gives the value of $x_i$.
The semantics of a command~$c$ is a function $\llbracket c\rrbracket : \mathbb{Z}^{n} \to \mathbb{Z}^{n}$,
transforming a valuation into another.
An SLP $\SLP = c_1,\dots,c_m$ defines the function~$\llbracket \SLP\rrbracket: \mathbb{Z}^{n} \to \mathbb{Z}^{n}$
obtained by composing the constituent commands: 
$\llbracket \SLP\rrbracket =  \llbracket c_m \rrbracket \circ \dots  \circ \llbracket  c_1 \rrbracket$.
Clearly this is a piecewise-affine function.
Given a function $f:\mathbb{Z}^{n} \to \mathbb{Z}^{n}$, we define its $m$-th
power as $f^m:\mathbb{Z}^{n} \to \mathbb{Z}^{n}$ where
\[
  f^m = \underbrace{f \circ \cdots \circ f}_{m \text{ times}}
\]
is the $m$-fold composition of~$f$. 

We denote by~$\mathcal{T}^{+}_n$ the set of terms~$a_1x_{j_1}+\cdots+a_nx_{j_n}+b$
where the coefficients $a_1,\cdots,a_n,b$ are in~$\{0,1\}$.  
An SLP that only uses terms in~$\mathcal{T}^{+}_n$ is called \emph{monotone}. 
Note that monotone SLPs induce monotone functions from $\mathbb{Z}^{n}$ to $\mathbb{Z}^{n}$ 
(subtraction and $\min$ are not allowed).

\subsection{The Powering Problem}
For an SLP~$\SLP$ of order~$n$, a  valuation $\valuation\in 
\mathbb{N}^n$
 and $m \in \mathbb{N}$ (encoded in binary), 
let $\valuation' = \llbracket \SLP \rrbracket^m(\valuation)$. 
Given two variables  $x,y\in var_n$ of the SLP, the \textbf{powering   problem}
asks whether 
$\valuation'(x) \geq \valuation'(y)$.
Since the initial valuations~$\valuation$ are always
non-negative,
 all valuations obtained by powering monotone SLPs are non-negative.
The above problem is  \P-complete 
 if the exponent~$m$  is written in unary~\cite{greenlaw1995limits}.

Observe that all numbers generated 
by powering an SLP can be represented using exponentially-many bits in the bitsize 
of the exponent. It follows that
 the powered SLP can be explicitly evaluated in exponential time. We provide a matching lower bound in Section~\ref{sec:maintech}. Before that, we show the connection of SLP powering to MDPs.

\subsection{Synchronized Reachability and SLP Powering}
The connection is stated in the following Theorem.

\begin{theorem}
\label{thm:power-to-sync}
The powering problem for monotone SLPs reduces, in polynomial time, to the finite-horizon synchronized reachability problem in MDPs.
\end{theorem}

\noindent
To illustrate this reduction, let us consider the SLP~$\SLP$
of order~$2$:
\[
	x_1 \leftarrow  \max(x_1+x_2, x_2+x_2);
    \quad \quad \quad \quad x_2 \leftarrow  		\max(x_1+x_1,x_1+x_1).
\]
This SLP is \emph{normalized}, that is to say all its $\max$ commands have exactly two arguments $t_1,t_2 \in \mathcal{T}^{+}_n$ and furthermore $t_1,t_2$ have exactly two summands.
 (Note that focusing on normalized SLPs is no loss of generality.) 
We are interested in the $2$-nd power of~$\SLP$ with initial valuation $\valuation(x_1)=0$ and $\valuation(x_2)=1$.
In Figure~\ref{fig:connection}, two copies of $\SLP$ are shown on  the right to 
visualize the concept of powering it.
To obtain an MDP, we consider a set of actions $A = \{a,b\}$ and have
each variable~$x_i$ become a state. In the example, $s$ and $t$ are the corresponding states for~$x_1$ and $x_2$.
The  $t_1,t_2$ arguments of $\max$ commands determine the successors of actions~$a,b$, respectively,
where each successor has probability~$\frac{1}{2}$.
The command
$x_1 \leftarrow \max(x_1+x_2, x_2+x_2)$ translates to $\prob(s,a)(s)=\prob(s,a)(t)=\frac{1}{2}$ and
$\prob(s,b)(s)=1$, as shown in the MDP in Figure~\ref{fig:connection}.
Since $\valuation(x_2)=1$, we make~$t$ a target state.
Now the $i$-th iteration of value iteration of \eqref{eq-rec-synch} (corresponding to the $i$-th step \emph{before} the horizon) is tightly connected to the $i$-th power of the SLP.
Indeed, letting $\valuation_i=\llbracket \SLP\rrbracket^i(\valuation)$,
one can prove that  $\Pv_{=i}(s)= \frac{1}{2^i}\valuation_i(s)$ and $\Pv_{=i}(t)=\frac{1}{2^i} \valuation_i(t)$.

\subparagraph{SLP vs. monotone SLP powering.} It thus remains to provide a
lower bound for the Monotone SLP powering problem. The crucial step, which we
cover in Section~\ref{sec:maintech}, is providing lower bounds for the
non-monotone variant. The remaining step from non-monotone to monotone
powering can be made by adapting the techniques of Allender et al.~\cite{AKM}.

\begin{theorem}
\label{thm:monotone-hard} The powering problem for arbitrary SLPs reduces, in polynomial time, to the powering problem for monotone SLPs.
\end{theorem}

\begin{figure}[t]
\begin{center}
\scalebox{.8}{
\begin{tikzpicture}
\draw [dashed, gray] (-.5-14,-.5) -- node[label= right:{\bf ~~~~~~~~~initialization}]{}  (2.5-14,-.5) ;
\node[rectangle,draw=none] at (0-14, -.25)   (s0) {$x_1:0$};
\node[rectangle,draw=none] at (2-14, -.25)   (t0) {$x_2:1$};
\node[rectangle, draw] at (0-14, 2)   (max0) {$max$};
\node[rectangle, draw] at (2-14, 2)   (dax0) {$max$};
\node[rectangle,draw] at (2-14, 1)   (xum0) {$+$};
\node[rectangle,draw] at (0-14, 1)   (sum0) {$+$};

\draw [dashed, gray ](-.5-14,2.5) -- node[label= right:{\bf ~~~~~~~~~$1$-st
power}]{}  (2.5-14,2.5) ;
\node[rectangle,draw=none] at (0-14, 2.75)   (s1) {$x_1$};
\node[rectangle,draw=none] at (2-14, 2.75)   (t1) {$x_2$};
\node[rectangle, draw] at (0-14, 5)   (max1) {$max$};
\node[rectangle, draw] at (2-14, 5)   (dax1) {$max$};
\node[rectangle,draw] at (2-14, 4)   (xum1) {$+$};
\node[rectangle,draw] at (0-14, 4)   (sum1) {$+$};

\draw [dashed, gray] (-.5-14,5.5) -- node[label= right:{\bf ~~~~~~~~~$2$-nd
power}]{}  (2.5-14,5.5) ;
\node[rectangle,draw=none] at (0-14, 5.75)   (s2) {$x_1$};
\node[rectangle,draw=none] at (2-14, 5.75)   (t2) {$x_2$};
%
%

\path[->,thick] (s0) edge (sum0);
\path[->,thick] (t0) edge (sum0);
\path[->,thick] (s0) edge[bend left] (xum0);
\path[->,thick] (s0) edge[bend right] (xum0);
\path[->,thick] (xum0) edge (max0);
\path[->,thick] (sum0) edge (max0);
\path[->,thick] (xum0) edge[bend left] (dax0);
\path[->,thick] (xum0) edge[bend right] (dax0);
\path[->,thick] (max0) edge (s1);
\path[->,thick] (dax0) edge (t1);

\path[->,thick] (s1) edge (sum1);
\path[->,thick] (t1) edge (sum1);
\path[->,thick] (s1) edge[bend left] (xum1);
\path[->,thick] (s1) edge[bend right] (xum1);
\path[->,thick] (xum1) edge (max1);
\path[->,thick] (sum1) edge (max1);
\path[->,thick] (xum1) edge[bend left] (dax1);
\path[->,thick] (xum1) edge[bend right] (dax1);
\path[->,thick] (max1) edge (s2);
\path[->,thick] (dax1) edge (t2);
%

\draw [dashed, gray] (-7.5,-.5) -- node[label= right:{\bf ~~~~~~~~~~$2$-nd
step}]{}  (-4.2,-.5) ;
\node[rectangle,draw=none] at (-7, -.25)   (s0) {$s$};
\node[rectangle,draw=none] at (-5, -.25)  (t0) {$t$};
\node[rectangle, draw] at (-7, 2)   (max0) {$?$};
\node[rectangle, draw] at (-5, 2)   (dax0) {$?$};
\node[rectangle,draw=none] at (-7.6, 1)   (xum0) {$a$};
\node[rectangle,draw=none] at (-6.4, 1)   (sum0) {$b$};
\node[rectangle,draw=none] at (-5, 1)   (dum0) {$a,b$};

\draw [dashed, gray] (-7.5,2.5) -- node[label= right:{\bf ~~~~~~~~~~$1$-st step}]{}  (-4.2,2.5) ;
\node[rectangle,draw=none] at (-7, 2.75)   (s1) {$s$};
\node[rectangle,draw=none] at (-5, 2.75)   (t1) {$t$};
\node[rectangle, draw] at (-7, 5)   (max1) {$?$};
\node[rectangle, draw] at (-5, 5)   (dax1) {$?$};
\node[rectangle,draw=none] at (-7.6, 4)   (xum1) {$a$};
\node[rectangle,draw=none] at (-6.4, 4)   (sum1) {$b$};
\node[rectangle,draw=none] at (-5, 4)   (dum1) {$a,b$};
\draw [dashed, gray] (-7.5,5.5) -- node[label= right:{\bf ~~~~~~~~~~starting step}]{}  (-4.2,5.5) ;
\node[rectangle,draw=none] at (-7, 5.75)   (s2) {$s$};
\node[rectangle,draw=none] at (-5, 5.75)   (t2) {$t$};

\path[->,thick] (xum0) edge node[near end,left]{$\frac{1}{2}$}(s0);
\path[->,thick] (xum0) edge node[near end,right]{$~~~\frac{1}{2}$}(t0);
\path[->,thick] (dum0) edge (s0);
\path[->,thick] (sum0) edge(s0);

\path[dashed,thick] (max0) edge (xum0);
\path[dashed,thick] (max0) edge (sum0);
\path[dashed,thick] (dax0) edge (dum0);
\path[dashed,thick] (s1) edge (max0);
\path[dashed,thick] (t1) edge (dax0);
\path[->,thick] (xum1) edge node[near end,left]{$\frac{1}{2}$}(s1);
\path[->,thick] (xum1) edge node[near end,right]{$~~~\frac{1}{2}$} (t1);
\path[->,thick] (dum1) edge (s1);
\path[->,thick] (sum1) edge(s1);

\path[dashed,thick] (max1) edge (xum1);
\path[dashed,thick] (max1) edge (sum1);
\path[dashed,thick] (dax1) edge (dum1);
\path[dashed,thick] (s2) edge (max1);
\path[dashed,thick] (t2) edge (dax1);

\node[draw=none] at (1, 1)   (s1) {MDP $\mdp$};
\node[state] at (0,2)   (s1) {$s$};
 \node[state] at (2,2)   (s4) {$t$};

\path[->] (s1)   edge  [loop above] node[midway, above, align=center] {$a:\frac{1}{2}$\\$b:1$}  (s1) ;
\path[->] (s1) edge [bend left] node[midway, above] {$a:\frac{1}{2}$} (s4) ;
\path[->] (s4) edge  [bend left] node[midway, below] {} (s1) ;

\end{tikzpicture}
}
\caption{ An example for the translation from SLPs to MDPs.
}
\label{fig:connection}
\end{center}

\end{figure}
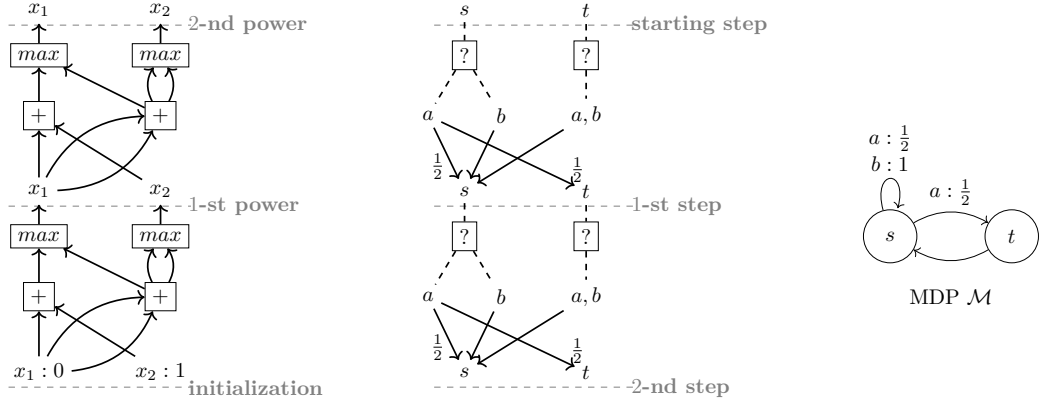

%% file: reduction.tex
\label{sec:maintech}

To show \EXP-hardness of all the problems introduced so far, we introduce a class of counter programs that allow 
linear updates on counters  and show that a (time-)bounded version of the termination problem 
for these  programs is  \EXP-complete.
Finally, we reduce this bounded termination problem to the powering problem.

A \emph{deterministic linear-update counter program} (CP)  consists of $n$
\emph{counters} $\{c_i \st 1 \leq i \leq n\}$, ranging over~$\mathbb{Z}$, and a
sequence of $m$ instructions. We consider instructions of the form
\[
    p : {}  c_1 \leftarrow c_2 + c_3 \quad \quad \quad 
    p : {}  \code{if } c_1 \geq c_2 \code{ goto } t \quad \quad \quad 
    p : {}  c_1 \leftarrow c_2 - c_3
\]
where $1 \leq p < m$ and $1 \leq t \leq m$, and the final instruction is always 
$m : \code{halt}$. 
More precisely, the instructions allow
\begin{enumerate}[(i)]
    \item adding or subtracting two counters, assigning the result to a third
        one, and continuing to the next
        instruction;
    \item  testing two counters against each other,
        and jumping to some given instruction if the result of the test is
        positive, continuing to the next instruction otherwise.
\end{enumerate}
The $\code{halt}$ instruction only loops to itself.

  A \emph{configuration} of a CP is  a tuple $(p,
v_1, \dots, v_n) \in \{1,\dots,m\} \times \mathbb{Z}^n$ consisting of an
instruction~$p$ and values of the counters (e.g., $v_1$ is the value for the
counter~$c_1$). We equip CPs with a fixed initial configuration lying in
$\{1\} \times \mathbb{N}^n$. Given a CP, the {\bf termination problem} asks whether
the $\code{halt}$ instruction is reached.  The {\bf bounded termination problem}
additionally takes as input an integer $N \in \mathbb{N}$, encoded in binary, and
asks whether the $\code{halt}$ instruction is reached within $N$ steps.

The bounded termination problem  is in \EXP:
 in a computation with~$N$ steps, the magnitude of the counters is
bounded by $2^N$, so each step can be simulated in time
exponential in the bitsize of~$N$.  We will now show that the problem is
\EXP-hard already for a certain subclass of  CPs
which facilitates  the reductions to  the powering problem.

\subparagraph{Simple counter programs.}
A CP is \emph{simple} if it satisfies the following conditions.  First, all
values in  all reachable configurations~$(p,v_1,\dots,v_n)$ are non-negative:
$v_i \in \mathbb{N}$ for all $1 \leq i \leq n$ (one may  ``guard''
subtractions by test instructions to achieve this).  Second,  all test
instructions
\(
    q : \text{\code{if} } c_i \geq c_j \text{ \code{goto} } r
\)
use counters $c_1$ and $c_2$ exclusively.  Moreover, for each such
instruction~$q$, there are counters $c_{\bar q_1},c_{\bar q_2}$ such that
in all reachable configurations~$(q,v_1,\dots,v_n)$ we have that 
\begin{enumerate}
    \item $v_1 = a_1 v_{\bar q_1}$ and $v_2 = a_2 v_{\bar q_2}$ with
        $a_1,a_2 \in \{64,64\cdot10,64\cdot12\}$. That is, the values of
        tested counters are ``scaled-up'' versions of the values of other
        counters.
    \item Additionally, the absolute difference of the values of the tested
        counters is larger than the values of all other counters, in symbols
        \(
            |v_1 - v_2| \geq \max\{v_k \st 3 \leq k \leq n \}.
        \)
\end{enumerate}
Note that the class of simple CPs is a semantically defined subclass of all
CPs. Further observe that for every test instruction we necessarily have that
$\bar q_1,\bar q_2 \geq 3$.

The following proposition kick-starts our sequence of reductions.
\begin{proposition}\label{pro:exp-hardness-simple-cms}
    The bounded termination problem for simple CPs is
    \EXP-complete.
\end{proposition}

To prove the proposition, we follow the classical recipe of first simulating a Turing machine
using a machine with two stacks, and then simulating the two-stack machine by a CP.
We note two key differences between our construction and the classical reduction:
($1$) We use the expressiveness of linear updates in CPs 
to simulate pushing and popping on the stack
in a linear
number of steps of the CP.
($2$) We instrument 
the two-stack machine to ensure that the height of the two stacks differs by
at most~$1$ along any computation. This is crucial  to allow us to
simulate the two-stack machine by a \emph{simple} linear-update counter program.


\subsection{From the Termination Problem to the Powering Problem}
\label{sec:slppow-min-hard}

We now sketch the main ideas behind the last (and most technically involved) missing link in our
sequence of reductions.

\begin{theorem}
\label{thm:counter-to-slppow}
The bounded termination problem for simple CPs reduces, in polynomial time, to the powering problem for SLPs.
\end{theorem}

\subparagraph{The encoding.}
Given a CP~$\calC$ we  construct an SLP $S$ of order $\geq 2n$
with variables including
$\{x_1,\dots, x_{2n}\}$. 
Let us denote $x_{n+i}$ by $Q_{i}$ for $1\leq i\leq n$.
The reduction is such that a configuration $(p, v_1, \dots, v_n)$ of~$\calC$ is
encoded as a valuation~$\nu:var_{2n} \to \mathbb{Z}$ of the SLP with the
property that
\(
    \nu(x_i)=v_i
\)
and
\(
    \nu(Q_i) = p \nu(x_i) = p v_i
\)
for all $1 \leq i \leq n$.
In this way, the instruction~$p$ of the CP is encoded in the variables of the SLP
(recall that SLPs are stateless). 

Given this encoding, the main challenge is to realize the transition function of
the CP as a function computed by an  SLP. Once this is accomplished, for every
$m\in \mathbb{N}$, the $m$-th power of the SLP~$S$ represents the $m$-step
transition function of the CP.

\subparagraph{Conditional commands.} 
Intuitively, to encode the transition function we would like to equip the SLP with
\emph{conditional commands}, whose execution depends on a conditional.
Specifically, we want to implement the following two kinds of conditional
updates
\[
    (y \leftarrow y \pm x_k \text{ if } Q_k = p x_k) \; \; 
    \text{ and } \; \;
    (Q_k \leftarrow p \cdot x_k \text{ if } x_i \geq x_j)
\] 
in terms of primitive commands of an SLP.  In both commands, if the condition is
not satisfied, the command is not executed, and the value of $y$ or $Q_k$
remains unchanged.
For example, one can simulate the first type of
conditional commands by executing
$y \leftarrow y \pm \max(0,x_k+t)$,
where $t$ is an expression that is $0$ if the test is
passed and less than $-x_k$ otherwise.
Intuitively, we think of $t$ as
``masking'' the assignment if the test fails. 

For the following result, which formalizes how we implement conditional
commands, we call a valuation~$\nu$ \emph{valid} if 
there exists $q\in\{1,\dots,m\}$
with
$\nu(x_i) \geq 0$ and
$\nu(Q_i) = q \nu(x_i)$ for all $1\leq i\leq n$. 

\begin{lemma}\label{lem:conditional-commands}
    Let $p\in \{1,\dots,m\}$ and  $i,j,k\in \{1,\dots, n\}$ be distinct.
    The following equation holds for all
    valid valuations~$\nu$:
    \begin{equation}\label{eq1}
        \max(0, \nu(x_k) + \min(\nu(Q_k) - p\nu(x_k),p\nu(x_k) - \nu(Q_k)))=
        \begin{cases}
            \nu(x_k) & \text{if } \nu(Q_k) = p \nu(x_k)\\
            0 &\text{otherwise.}
        \end{cases}
    \end{equation}
    Moreover, if  $\abs{\nu(x_i)-\nu(x_j)} \geq \nu(x_k)$, 
    then the following holds:
    \begin{equation}\label{eq2}
        \max(0, \nu(x_k) + \min(0, \nu(x_i) - \nu(x_j))) = 
        \begin{cases}
            \nu(x_k) & \text{if } \nu(x_i) \geq \nu(x_j) \\
            0 &\text{otherwise.} 
        \end{cases}
    \end{equation}
\end{lemma}
\begin{proof}
    The  equations  follow directly  from the assumption that  $\nu$ is valid,
    since  if $\nu(Q_j)\neq p\nu(x_j)$ then we also have
    $\abs{\nu(Q_j)-p\nu(x_j)} \geq \nu(x_j)$. In addition, if
    $\abs{\nu(x_i)-\nu(x_j)} \geq \nu(x_k)$ and $\nu(x_i) < \nu(x_j)$,
    we will have $\nu(x_k) + \nu(x_i) - \nu(x_j) \leq
    0$.
\end{proof}

Using the property that the simulated program is simple,
Equation~\eqref{eq1} can be used to simulate the conditional update $(y
\leftarrow y \pm x_k \text{ if } Q_k = p x_k)$ where $t=\min(Q_k - px_k,px_k -
Q_k)$ masks the update. Likewise, Equation~\eqref{eq2} can be used to simulate  the second type of
conditional update $(Q_k \leftarrow p \cdot x_k \text{ if } x_i \geq x_j)$ where
the masking expression is $t = \min(0,x_i - x_j)$.
Finally, the multiplication-by-a-constant required for the second type 
of the conditional update is achieved via repeated addition.

\subparagraph{Encoding the instructions.}
We recall that we encode being at the instruction~$p$ of the CP
by a valuation $\nu$ such that $\nu(Q_i)=p\nu(x_i)$ for all $1\leq i\leq n$.

Using the aforementioned conditional commands, we can
construct the SLP $S$ as the 
composition of $m$
smaller SLPs.
Each sub-SLP  $\pi_p$ 
simulates an instruction $p$
from the given
CP $\calC$.
Hence $S$,
when applied upon a valid valuation $\nu$ 
(i.e., a properly-encoded configuration
of $\calC$), simulates all of its instructions 
at once. 
By using conditional commands,
we make sure that
only one sub-SLP results in 
a non-zero update: executing $\pi_p$
has no effect on the valuation unless $p\nu(x_i)=\nu(Q_i)$
for all~$1 \leq i \leq n$.

In this way, powering $S$ allows us to simulate
consecutive steps of $\calC$. In particular, for all $N \in \mathbb{N}$ we have that
\(
    \llbracket S \rrbracket^N(\nu)(Q_1)
    \geq 
    m \cdot \llbracket S \rrbracket^N(\nu)(x_1),
\)
where $m$ is the $halt$ instruction, holds
if and only if $\calC$ halts after at most $N$ steps.

\section{Conclusion}

By the virtue of our chain of reductions (see Figure~\ref{fig:chain}), we get the following theorem.

\begin{theorem}
\label{thm:summary}
All the following problems are \EXP-complete:
\begin{itemize}
\item The finite-horizon reward problem for MDPs, and thus also the \textsc{ValIt} problem.
\item The finite-horizon reachability and synchronized reachability problems for MDPs.
\item The powering problem for SLPs and for monotone SLPs.
\item The bounded termination problem for simple counter programs.
\end{itemize}
\end{theorem}

The exact complexity of the following variant of the problem remains
open: given an MDP and a horizon encoded in binary,
determine whether there exists a policy achieving some given
expected-reward threshold (with no restriction on the actions used to do so).

%% file: appendix.tex
\section{Proof of Theorem~\ref{th-sync-to-rew}}
\label{app:sync-to-rew}
In this section, we will show that there is a polynomial-time many-one reduction from the finite-horizon synchronized-reachability problem to the finite-horizon reward problem --- for any discount factor $\gamma$.
\begin{proof}[Proof of Theorem~\ref{th-sync-to-rew}]
    Consider an instance of the finite-horizon synchronized-reachability
    problem, i.e., an MDP
    $\mdp=(\states,\actions,P,\cdot,\gamma)$, an initial state $s_0$, a target vertex $t$, and horizon $H$.
    In polynomial time we
    can construct a new MDP $\mdp'=(\states',\actions,P',R,\gamma)$ with
    an initial state $s'_0$ and horizon $H'$ such that:
    \begin{itemize}
        \item $\states' = \states\times\{\even,\odd\}$;
        \item $P'$ is such that for each $s,s'\in\states$ and
            $a\in\actions$ we have:
            \(
            P'((s,\even),a)(s',\odd) = P(s,a)(s')
            \)
            and
            \(
            P'((s,\odd),a)(s,\even) = 1;
            \)
            for all other arguments, $P'(\cdot,\cdot)(\cdot)$ returns zero;
        \item finally, $s_0' = (s_0,\even)$.
    \end{itemize}
    For $\mdp'$ we construct (also in polynomial time), the reward function 
    $\rewards$ as follows: for each $(s,a)\in \states\times\actions$ such that 
    $s\neq t$ we have $\rewards((s,\even),a)=0$ and
    $\rewards((s,\odd),a)=1/\gamma$. If $s = t$,
    then $\rewards((s,\even),a)=1$ and $\rewards((s,\odd),a)=0$. This ensures 
    that under any strategy $\sigma'$ in $\mdp'$ the
    $i$-step (expected) value, which we denote by 
    $\Rv_{i,\sigma'}(s_0)$,
    is equal to
    \[
    	\frac{(1-\gamma^i)}{1 - \gamma^2}
    \]
    if $i$ is even, 
    and otherwise it is equal to 
    \[
        \frac{(1-\gamma^{i-1})}{1 - \gamma^2}
        + \gamma^{i-1} \cdot
        \probm_{s'_0,\sigma'}(\{s'_0,\dots,s'_{i-1}:
        s'_{i-1} = t'\}),
    \]
    where $t'=(t,\even)$.  Finally, we set $H' = 2H + 1$.

    Intuitively, $\mdp'$ is formed by subdividing each probabilistic transition
    into two transitions by using newly added ``middle'' states (those of the
    form $\states\times\{\odd \}$). Hence, there is a one-to-one correspondence
    between runs of some length $\ell$ in $\mdp$ and finite paths of
    length $2 \ell$ in $\mdp'$. The correspondence naturally extends to 
    sets of runs and strategies (since there is no real choice in states of the
    form $\states\times\{\odd \}$). Moreover, under corresponding strategies,
    the probabilities of corresponding sets of runs are identical.
    Hence, for all $0 \leq p \leq 1$ and all $j \in \mathbb{N}$
    there exists a strategy $\sigma$ in $\mdp$ such that
    $\probm_{s_0,\sigma}(\{s_0,\dots,s_j:s_j =t\}) \geq p$ in $\mdp$
    if and only if there is a
    strategy $\sigma'$ in $\mdp'$ such that
    $\probm_{s'_0,\sigma'}(\{s'_0,\dots,s_{2j}:s_{2j} = t'\}) \geq p$
    in $\mdp'$.
    But by the discussion in
    the previous paragraph, such a strategy exists if and only if there is
    $\sigma''$ in $\mathcal{M}'$ such that 
    \[
        \Rv_{2j+1,\sigma''}(s'_0) \geq
        \frac{(1-\gamma^{2j})}{1 - \gamma^2}
        + \gamma^{2j} \cdot p.
    \]
    It follows that an action $a$ is an optimal first action for the finite-horizon synchronized-reachability problem in $\mathcal{M}$ (with horizon $H$) if and only if it is an optimal first action for the finite-horizon reward problem in $\mathcal{M}'$ (with horizon $H'$).
\end{proof}

\section{Proof of Theorem~\ref{th-finite-reach}}
\label{app:finite-reach}

We show that there is a polynomial-time many-one reduction from the
    finite-horizon synchronized-reachability problem to the
    finite-horizon reachability problem.

\begin{proof}[Proof of Theorem~\ref{th-finite-reach}]
    Given an instance $(\mdp,t,s_0,H)$ of the finite-horizon
    synchronized-reachability problem  --- where $t$ is the target state,
    $s_0$ is an initial state, and $H$ is a binary-encoded horizon ---  we
    compute an instance $(\mdp',\{g\},s_0,H+1)$ of the finite-horizon
    reachability problem where the
    MDP~$\mdp'$ is a modification of~$\mdp$. In particular, $g,b$ (standing for
    ``good'' and ``bad'', respectively) are two fresh states, both sinks
    in~$\mdp'$.  For all other states and all actions~$a$ we define a modified
    transition function~$P'$: with probability~$\frac{1}{3}$ action~$a$
    leads to the new target~$g$, and with probability~$\frac{2}{3}$ it does
    whatever it did in~$\mdp$; formally, we put $P'(s,a)(g)=
    \frac{1}{3}$ and $P'(s,a)(s')=\frac{2}{3} P(s,a)(s')$
    for all states $s,s'$ in~$\mdp$. Finally, we introduce a fresh action, $f$,
    which leads to $g$ and $b$, both with probability $\frac{1}{2}$, when played
    from~$t$; formally, we put $P'(t,f)(g) = P'(t,f)(b) =
    \frac{1}{2}$ and $P'(s,f)(b) = 1$ for all $s
    \neq t$. Observe that from state~$t$ and at time~$H$ it is best to
    play~$f$, as $\frac{1}{2} > \frac{1}{3}$. Strictly before time~$H$, any
    optimal strategy does not play~$f$ --- regardless of the current state ---
    as playing two different consecutive actions instead leads to~$g$ with
    probability at least $1 - (\frac{2}{3})^2 = \frac{5}{9} > \frac{1}{2}$;
    rather, any optimal strategy maximizes the probability of reaching
    state~$t$ at time~$H$ (because then it is beneficial to play~$f$). It
    follows that any optimal strategy (for synchronization) in~$\mdp$
    corresponds naturally to an optimal strategy (for reachability)
    in~$\mdp'$, and vice versa.  Specifically, the first action played in the
    two strategies is the same.  This gives the reduction.
\end{proof}

\section{Proof of Theorem~\ref{thm:monotone-hard}}\label{app:monotone-hard}

We  show how to remove all but one last subtraction --- which is then subsumed by the comparison --- in a Max-Plus-Minus SLP.

\begin{proof}[Proof of Theorem~\ref{thm:monotone-hard}]
 Let $S$ be a (general) SLP of order $n$. Without loss of generality, 
we suppose $S$ consists only of (binary) addition, subtraction, and $\max$ commands.
    To eliminate subtractions, we closely follow the proof of Theorem $8$ in
    Allender et al.~\cite{AKM}. We construct a monotone SLP $S'$ of order $n+1$
    from $S$ by first introducing a new variable $z$
    that will help us maintain the invariant that for all $1 \leq i \leq n$
    and all $1 \leq m
    \leq N$ we have $\llbracket S \rrbracket^m(\nu)(x_i) = \llbracket S'
    \rrbracket^m(\nu)(x_i) - \llbracket S' \rrbracket^m(\nu)(z)$.

    We proceed as follows: First we initialize $z \leftarrow 0$. Whenever we
    encounter a command $x_k \leftarrow x_i \pm x_j$ in $S$, we replace this
    with the following sequence of commands in $S'$.
    \par\medskip\noindent
    \begin{minipage}{0.4\textwidth}
        \begin{equation}\label{eqn:add-coms}\tag{Add}
        \begin{aligned}
            x_k &\leftarrow x_i + x_j \\
            x_{\ell} &\leftarrow x_{\ell} + z,  \ \forall \ell \neq k \\
            z &\leftarrow z + z
        \end{aligned}
        \end{equation}
    \end{minipage}
    \hfill
    \begin{minipage}{0.4\textwidth}
        \begin{equation}\label{eqn:sub-coms}\tag{Sub}
        \begin{aligned}
            x_k &\leftarrow x_i + z \\
            z &\leftarrow z + x_j \\
            x_{\ell} &\leftarrow x_{\ell} + x_j, \ \forall \ell \neq j,k\\
            x_j &\leftarrow x_j + x_j
        \end{aligned}
        \end{equation}
    \end{minipage}
    \par\medskip\noindent
    For addition commands, i.e. if $x_k \leftarrow x_i + x_j$, replace
    it by the sequence of commands~\eqref{eqn:add-coms}; for subtraction
    commands, i.e. if $x_k \leftarrow x_i - x_j$, by the sequence of
    commands~\eqref{eqn:sub-coms}.

    We leave the $\max$ commands in the SLP $S$ unchanged in $S'$. Notice that
    the sequence of commands above are monotone (i.e. they include no
    subtraction commands), and the number of commands used in $S'$ is at most
    $(n+3)$ times the number of commands used in $S$ plus an additional command
    to initialize $z$. To complete the proof, it suffices to note that for
    every command in $S$, the corresponding sequence of $\{\max,+\}$ commands in
    $S'$ satisfy the invariant that the value of every variable $x_i$ for all $1
    \leq i \leq n$ of $S$,  is obtained as the difference of variables $x_i$ and
    $z$ in $S'$.
\end{proof}

\section{Proof of Proposition~\ref{pro:exp-hardness-simple-cms}}
\label{app:CP-hard}

\begin{figure}[t]
\begin{center}

\begin{tikzpicture}[stack/.style={rectangle split,
rectangle split parts=#1,draw, anchor=center}]

\node[stack=7] at (0,0) { \code{110}
\nodepart{two} \code{101}
\nodepart{three}\code{110}
\nodepart{four} \code{110}
\nodepart{five} \code{100}
\nodepart{six} \code{100}
\nodepart{seven} $\times$
};
\node[draw=none] at (0,-2){counter~$c_3$ storing stack~$s_1$};
\draw[<->,dotted] (-.5,1.6) -- 
  node[rotate=90,above] {height $=7$} (-.5,-1.6);

\node[stack=7] at (5,0) { \code{110}
\nodepart{two} \code{110}
\nodepart{three}\code{100}
\nodepart{four} \code{100}
\nodepart{five} $\times$
\nodepart{six} $\times$
\nodepart{seven} $\times$
};
\node [draw=none] at (5,-2){counter~$c_4$ storing stack~$s_2$};
\draw[<->,dotted] (5.5,1.6) --
  node[rotate=90,below] {height $=7$} (5.5,-1.6);

\node[draw=none] at (2.5,1.2){pointer~$\ell$};
\draw[<-](.5,1.2)--(1.7,1.2);
\draw[<-](4.5,1.2)--(3.3,1.2);

\end{tikzpicture}

\caption{The stack $s_1$ contains the word $abaa\square \square$
    and stack $s_2$ 
	contains the word~$ba\square \square$. If the state of the Turing machine is
	q then the overall configuration of the
    Turing machine is $\square \square aaba\langle q \rangle ba \square \square$.
    Here, $\ell=8^{6}$ as the height of both stacks 
	is $7$. 
}
\label{fig:general}
\end{center}
\end{figure}
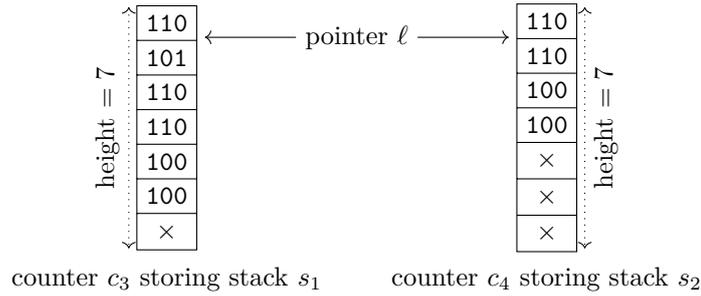
The proof is by a reduction from the termination problem of a 
$\mathcal{O}(2^n)$-time bounded Turing machine.
We suppose the tape alphabet of the Turing machine is $\{a,b,\square\}$ with
$\square$ denoting the empty-cell symbol. For convenience, we also assume
that the machine always overwrites empty-cell symbols it reads and replaces
them with some sequence of symbols meant to internally represent an
empty-and-read cell. This is clearly no loss of generality. Moreover
it implies that cells with $a$'s and $b$'s are never
separated by $\square$'s.
    
    The tape can be encoded into the
    two stacks $s_1$ and $s_2$, such that the top of the stack $s_1$ encodes the
    contents of the tape cell the head of the Turing machine is currently at;
    the top of $s_2$, those of the tape cell immediately right of it; the bottom
    of $s_1$, the leftmost part of the simulated tape; the bottom of $s_2$, its
    rightmost part.  The stack alphabet we use consists of $3$ binary codes
    \code{110}, \code{101}, \code{100}, corresponding to the symbols from the
    Turing-machine tape $a$, $b$, $\square$.
    We then keep the contents of $s_1$ and $s_2$ in counters $c_3$ and $c_4$ by
    interpreting the binary codes as \emph{little-endian} binary numbers (that
    is, we write bits from left to right and the right-most bits are the least
    significant) and keeping the tops of the stacks in the most significant bits
    of the counters.  The technical difficulty of the proof lies in making sure
    that the resulting counter program is simple and in simulating \code{push}
    and \code{pop} instructions efficiently by using linear updates on the
    counters. Apart from counters $c_3,c_4$ used to simulate the stacks as previously
    described, we use counters $c_1,c_2$ for tests as well as additional counters described below.

    Let $i \in \{1,2\}$.
    The contents of the stack $s_i$ (encoded by counter $c_{2+i}$) end after the
    first \emph{end-of-stack} symbol \code{100}. That is, since \code{100}
    represents an empty cell in the Turing-machine tape, the rest of the tape
    (and stack) contents are considered to be trash. Note that this means our
    stacks have a \emph{virtual height}, the number of symbols before an
    end-of-stack symbol, and an actual height, the actual number of symbols.
    Throughout our simulation, we will make sure both stacks always have the
    same height $h$. Additionally, we will keep a counter $\ell$ whose value
    will always be a power of $8$ divided by $2$. More precisely, we will maintain the
    invariant $\ell = 8^{h}/2$. Intuitively, $\ell$ will be used as a pointer to
    the first bit (from left to right, thus its most significant bit) of the
    second element (from top to bottom) of both stacks as encoded in the
    corresponding counters. Finally, we will also make sure $h \ge 2$ always
    holds by initializing both stacks with two \code{100} symbols that will
    never be popped. (Additionally, all the \emph{trash}
    below these two symbols in the stack will be just
    invalid \text{000} codes.)
    See Figure~\ref{fig:general} for an illustration of the
    proposed encoding. Henceforth, for all $j$,
    we denote by $v_j$ the value of the counter
    $c_j$.
    Observe that the above definitions imply the following
    invariant always holds.
    \begin{equation}\label{inv:l-pointer}
        8\ell < v_{2+i} < 16\ell \text{ and }
        v_{2+i} \not\in \{8\ell, 10\ell, 12\ell\}
    \end{equation}
    Indeed $8\ell$, $10\ell$, and $12\ell$, can be thought of as simulated
    stacks with $\square$, $b$, and $a$, as their respective topmost symbols and
    invalid \code{000} symbols below.
    Therefore, one can test the symbol
    $\sigma$ at the top of a simulated stack $s_i$ by using the corresponding
    counter test based on the following equivalences.
    \begin{equation}\label{eqn:equiv-tests}
    \begin{aligned}
        \sigma = a \iff & 12\ell \leq v_{2+i} < 14\ell\\
        \sigma = b \iff & 10\ell \leq v_{2+i} < 12\ell\\
        \sigma = \square \iff & 8\ell \leq v_{2+i} < 10\ell
    \end{aligned}
    \end{equation}
    Unfortunately, doing as described above may not always yield a simple
    linear-update counter program. In the sequel we will describe how to
    simulate pushing and popping symbols from the stacks as well as how to test
    popped symbols to obtain a simple counter program.
    
    \subsection{Pushing} To push a symbol $\sigma$ into $s_1$, the counter
    program executes the following instructions: $\ell \leftarrow 8\ell$; $c_3
    \leftarrow c_3 + s\ell$; $c_4 \leftarrow 8c_4$; where $s$ is $8$, $10$, or
    $12$, depending on whether $\sigma$ is $\square$, $b$, or $a$, respectively.
    Note that multiplication is not an allowed operation for our linear-update
    counter programs, so, for instance, the instruction $\ell \leftarrow 8\ell$
    has to be realized via a (constant) number of addition instructions.
    Roughly speaking, three new most-significant bits are ``freed'' in $c_3$ and
    then set by adding to it a stack whose topmost symbol is the one we want to
    push. Then, we left-shift (multiply by $8$, to be technically correct) the
    contents of the other counter since we have updated the height pointer.  The
    operations are executed \emph{mutatis mutandis} when pushing a symbol into
    $s_2$.
    
    \subsection{Popping and testing} To pop a symbol from a stack, we set to
    zero the $3$ most significant bits of the counter simulating the popped
    stack and then left-shift it. Additionally, prior to popping, we test the
    symbol at the top of the stack. Once more, let $i \in \{1,2\}$. Instead of
    using the equivalences from Equation~\eqref{eqn:equiv-tests}, we observe
    that thanks to the invariant from Equation~\eqref{inv:l-pointer} and because
    of our encoding we have that
    \begin{equation}\label{eqn:simpleness}
    \begin{aligned}
        \max\{v_3,v_4,\ell\} & \leq 16\ell \leq 64|v_{2+i} - 12\ell|, \text{ and}\\
        \max\{v_3,v_4,\ell\} & \leq 16\ell \leq 64|v_{2+i} - 10\ell|.
    \end{aligned}
    \end{equation}
	The left inequalities follow from the fact that
    $v_{2+i} < 16\ell$ and $\ell$ is positive.
    For the right inequalities we have the following lemma.
    \begin{lemma}
        For all $i \in \{1,2\}$ and both $t \in \{10,12\}$ it holds that
        \(
            \ell \leq 4 |v_{2+i} - t\ell|.
        \)
    \end{lemma}
    \begin{proof}
		%
        Since all codes have a $1$ as their most significant bit, and since the
        symbol below the top of the stack is always $a$, $b$, or $\square$,
        (recall that we push two $\square$ into the stack from the beginning and
        never pop them) then we have the following
        refinement of Equation~\eqref{eqn:equiv-tests}
        where $\sigma$ denotes the top of the simulated
        stack.
        \begin{equation}\label{eqn:new-invariant}
        \begin{aligned}
        \sigma = a \iff & 13\ell \leq v_{2+i} < 14\ell\\
        \sigma = b \iff & 11\ell \leq v_{2+i} < 12\ell\\
        \sigma = \square \iff & 9\ell \leq v_{2+i} < 10\ell
    	\end{aligned}
        \end{equation}
        It follows that the claim holds whenever $v_{2+i} > t\ell$.
Observe that by the above inequalities $v_{2+i}$ is never $t\ell$.

        Let us now focus on the cases when $v_{2+i} < t \ell$. This can only
        happen if $b$ is on the top of the stack and $t = 12$ or if $\square$ is
        on the top of the stack and $t \in \{10,12\}$.
        Note that it suffices to consider the cases when $v_{2+i}$ encodes a
        stack with topmost symbol $b$ or $\square$ and $t$ is $12$ or $10$,
        respectively.
        We have the
        following equivalences
        \begin{align*}
            \ell =& \left(\overbrace{\code{000}\, \code{100}\, \code{000}\,
        \dots}^{h \text{ blocks of } 3} \right)_2\\
            10\ell =& \left(\code{101}\, \code{100}\, \code{000}\,
                \dots\right)_2\\
            12\ell =& \left(\code{110}\, \code{100}\, \code{000}\,
                \dots\right)_2
        \end{align*}
        Additionally, because of our encoding (see 
        Equation~\eqref{eqn:new-invariant}) we also have the following.
        \begin{align}
            \label{eqn:symb}
            v_{2+i} \leq & \left(\overbrace{\code{101}\, \code{110}\, \code{110}\,
                \dots}^{h = \log_8(2\ell)} \right)_2 &
                \text{if the top symbol is } b\\
            \label{eqn:empty}
            v_{2+i} = & \left(\code{100}\, \code{100}\, \code{000}\, \dots\right)_2
                & \text{if the top symbol is } \square
        \end{align}

        For the first case, i.e. $v_{2+i}$ encodes a stack with $b$ as the
        symbol in its top and $t = 12$. For Inequality~\eqref{eqn:symb}, on the right-hand side
        we present the largest integer realizable via
        our encoding, i.e., the contents of the stack are
        (from top to bottom) $b a a \dots$ We let $A$ denote
        the right-hand side of the Inequality~\eqref{eqn:symb}. Observe that
        \[
            A + \left(\overbrace{
                \code{000}\, \code{001}\, \code{000}\, \dots
            }^h \right)_2 \leq 12\ell.
        \]
        The desired result thus follows from Inequality~\eqref{eqn:symb} and the
        fact that the second summand in the above equation is $\frac{\ell}{4}$.
        For the second case, i.e. $v_{2+i}$ encodes a stack with $\square$ on
        its top and $t = 10$, we proceed similarly. It suffices to
        note that $v_{2+i} +
        \ell \leq 10\ell$.
    \end{proof}

    It follows that we can copy $64v_{2+i}$ and $64t\ell$ (for $t \in
    \{10,12\}$) into 
    counters $c_1$
    and $c_2$ and test them instead of $c_{2+i}$ and $\ell$ directly. More
    precisely, we execute the following
    \begin{align*}
        {} : {} & c_1 \leftarrow 64c_{2+i}\\
        {} : {} & c_2 \leftarrow 64\cdot 12 \cdot \ell\\
        {} : {} & \code{if } c_1 \geq c_2 \code{ goto } p\\
        {} : {} & c_2 \leftarrow 64\cdot 10 \cdot \ell\\
        {} : {} & \code{if } c_1 \geq c_2 \code{ goto } q\\
        r : {} & \dots
    \end{align*}
    where from instruction $p$ onward we handle the case 
    where the topmost
    symbol is \code{110} (encoding an $a$); from $q$, the case 
    where it is
    \code{101} (encoding a $b$); and from $r$, the case where 
    it is \code{100}
    (encoding $\square$). From $p$, we first make sure the 
    \code{pop} is
    correctly simulated by executing \( p : c_{2+i} \leftarrow 
    c_{2+i} -
    12\ell;~c_{2+i} \leftarrow 8c_{2+i}.  \) Popping $b$ from 
    $q$ is handled
    similarly.
    
    \subsection{Simplicity}
    To conclude the proof, we observe that the above instructions --- in
    particular, the ones simulating the popping and testing --- satisfy the
    properties required of the tests for the program to be simple (see
    Equation~\eqref{eqn:simpleness}). Note that on faithful simulations of a
    Turing machine, all reachable configurations of the constructed counter program have
    non-negative counter values. Hence, the program is indeed simple.
    \qed

\section{Proof of~\texorpdfstring{\Cref{thm:counter-to-slppow}}{Theorem~\ref{thm:counter-to-slppow}}\label{app:SLP-hard}}

In this section, we show that the Powering problem is \EXP-complete. To this end, we reduce the bounded termination problem for simple linear-update counter programs to the \emph{powering problem}. 
    
   Let us consider a linear-update counter program $\calC$ with $n$ counters and
    $m$ instructions. Without loss of generality, we will assume that the
    addition and subtraction instructions are all of the form $c_i \leftarrow
    c_i \op c_j$. That is, the assigned counter $c_i$ appears as the left
    operand of the instruction. Indeed, for any instruction $c_i \leftarrow c_j
    \op c_k$ that does not conform to this (i.e. $i \neq j$), we can replace it
    by a sequence of three conforming instructions that achieves the same
    effect: a ``reset'' via subtraction $c_i \leftarrow c_i - c_i$; a ``copy''
    via addition $c_i \leftarrow c_i + c_j$; and, finally, the intended
    operation $c_i \leftarrow c_i \op c_k$.  Since no test instruction is
    introduced, this transformation preserves simplicity of the program. 
    Recall that the halt instruction, i.e. the $m$-th instruction
	of $\calC$, is of the form $m : \code{if } c_1 \geq c_1
	\code{ goto } m$. That is, $m$ is ``absorbing''.
    
    We will now describe how the Max-Plus-Minus SLP $S$ is constructed. The
    SLP has two distinct variables $z,o$ whose values will be initially
    set to $0$ and $1$, respectively. Additionally, we will make sure that 
    the values of
    these variables remain constant with respect to applications of $S$. 
    For convenience, we will use multiplication by integer constants in the 
    description of the SLP. These are to be interpreted as sequences
    of addition commands. (We can obtain multiplication by
    binary constants using polynomially-repeated doubling.) Apart from the variables $x_1,\dots,x_n,Q_1,\dots,Q_n$ introduced in the main text, we use auxiliary variables $d_i,d_i',e_i,s_i,s_i',s_i''$ for each $1\leq i \leq n$. 
    
    \subsection{Linear updates} First, for all
    $1 \leq i \leq n$ we initialize variables $d_i$
    by appending to $S$ the following command
    \[
    	d_i \leftarrow z.
    \]
    Now, for all addition and subtraction instructions
    \[
    	p : c_i \leftarrow c_i \op c_j
    \]
    from $\calC$ we append to the SLP the following commands
    \[
    	d'_i \leftarrow \begin{cases}
        	x_j & \text{if } Q_j = p x_j\\
            0 & \text{otherwise}
        \end{cases}
    \]
    and
    \begin{gather*}
    	d_i \leftarrow d_i \op d'_i\\
        e_i \leftarrow (p+1) \cdot d_i.
    \end{gather*}
    Note that the first command is a conditional command, and 
    so we are making
    use of its implementation from
    Lemma~\ref{lem:conditional-commands} (recall
    that $\calC$ is simple, and thus the lemma indeed
    applies).
    
    Observe that if the $d_i$ are not updated by commands in the rest of the SLP, then
    we have that, for all $1 \leq i \leq n$,
    $d_i$ is non-zero only if $\nu(Q_j) = p \nu(x_j)$.
    Hence, if the valuation of the variables
    before the application of $P$ is a correct encoding of a configuration of $\calC$
    then, for all $1 \leq i \leq n$, $d_i$ is non-zero if and only if the current
    simulated instruction will affect the value of $c_i$ by
    a non-zero value of $d_i$.
    
    \subsection{Tests and jumps} For all test-and-jump instructions
    \[
    	p : \code{if } c_i \geq c_j \code{ goto } q
    \]
    from $\calC$ we append, for all $1 \leq k \leq n$, the following commands
    to $S$
    \begin{gather*}
    	s''_k \leftarrow \begin{cases}
        	x_k & \text{if } x_i \geq x_j\\
            0 & \text{otherwise,}
        \end{cases}\\
    	s'_k \leftarrow \begin{cases}
        	s''_k & \text{if } Q_j = p x_j\\
            0 & \text{otherwise,}
        \end{cases}
    \end{gather*}
    and
    \[
    	s_k \leftarrow s'_k \cdot (q - p - 1).
    \]    
    If $i = j$, then we forego the first conditional
    commands and just set $s''_k$ to the value of
    $x_k$.
    
    Some remarks about the above commands are in order. First, the command
    updating $s''_k$ is not completely covered by
    Lemma~\ref{lem:conditional-commands} since we may have $k \in \{i,j\}$.
    The second command is also not explicitly
    implemented in Lemma~\ref{lem:conditional-commands}.
    It should, however, be easy to see that the second command is indeed implementable using
    our chosen encoding using the same Max-Plus-Min
    expression used for the first part of Lemma~\ref{lem:conditional-commands} (one possible way is to use similar approach as for  	modifying the addition in counter programs above, i.e. to first reset $s'_k$ using subtraction and then use addition to perform the copy).
    Let us now extend Lemma~\ref{lem:conditional-commands} to cover the first
    command.

\begin{lemma}
    Let $p\in \{1,\dots,m\}$ and $i,j,k\in \{1,2\}$ be such that $i \neq j$ and
    $k \in \{i,j\}$. The following equation holds for all valid valuations~$\nu$
    that encode configurations of a simple linear-update CM, i.e.
    $\nu(x_i)=v_i$ and    $\valuation(Q_i) = p\valuation(x_i)$ for all $1 \leq i
    \leq n$, before executing a test-and-jump instruction $q : \code{if } c_i
    \geq c_j \code{ goto } r$.
    \begin{equation}\label{eqn:eq2}
        a_k \cdot \max(0, \nu(x_{\bar q_k}) + \min(0, \nu(x_i) - \nu(x_j))) = 
        \begin{cases}
            \nu(x_k) & \text{if } \nu(x_i) \geq \nu(x_j) \\
            0 &\text{otherwise.} 
        \end{cases}
    \end{equation}
(Here $a_k$ is as defined in~\Cref{sec:reductions}.)
\end{lemma}
\begin{proof}
	By simplicity of the CM, and since $k \in \{i,j\}$,
    we know that there exists 
    $x_{\bar q_k }$ such that $\nu(x_k) = a_k 
    \nu(x_{\bar q_k})$, i.e. $x_k$ is the 
    $a_k$-scaled-up version of $x_{\bar q_k}$.
    Hence, it suffices to show that the left-hand side of
    Equation~\eqref{eqn:eq2}
    evaluates to $a_k \nu(x_{\bar q_k})$ if $\nu(x_i) \geq 
    \nu(x_j)$ and to $0$ otherwise.
    
    We have that $\nu(x_{\bar q_k})
    + \nu(x_i) - \nu(x_j) \leq 0$ 
    whenever
    $\nu(x_i) < \nu(x_j)$
    by simplicity since
    $\abs{\nu(x_j) - \nu(x_i)} \geq x_\ell$ for all
    $\ell \neq i,j$ and in particular for $\ell = 
    {\bar q_k}$. 
    Hence, the outermost $\max$ yields a positive
    number only if $\nu(x_i) > \nu(x_j)$ and $0$ otherwise. 	Furthermore,
    it yields $\nu(x_{\bar q_k})$ because the
    inner $\min$ evaluates to $0$ if $\nu(x_i) \geq \nu(x_j)$
    does indeed hold.
\end{proof}
    The appended commands are such that, if the $s_k$
    are not updated by commands in the rest of the SLP, then we have that,
    for all $1 \leq k \leq n$, $s_k$ is non-zero only if $\nu(Q_j) = p \nu(x_j)$
    and $\nu(x_i) \geq \nu(x_j)$. It follows that if the valuation of the variables
    before applying $S$ satisfies our encoding then $s_k$ holds the ``jump difference''
    (precisely, the difference $q - p - 1$ if $v_i \geq v_j$
    holds for the simulated configuration and $0$ otherwise)
    that should be applied to all $Q_i$ so as to obtain the new correct encoding
    for the next configuration of $\calC$. Note that, we have overshot by $-1$ if
    the test of the simulated instruction is passed. This will be compensated
    for in the sequel.
    
    \subsection{Final counter and PC update} To conclude our construction of $S$
    we have to affect the $x_i$ by the $d_i$ and update the $Q_i$ to account
    for the change of configuration in $\calC$. We can do the latter via
    the following commands appended for all $1 \leq i \leq n$.
    \begin{gather*}
    	Q_i \leftarrow Q_i + x_i\\
    	x_i \leftarrow x_i + d_i\\
        Q_i \leftarrow Q_i + e_i\\
        Q_i \leftarrow Q_i + s_i
    \end{gather*}
    Note that the first command increases the 
    ``program counter'' by exactly one, i.e.
    if $\nu(Q_i) = p \nu(x_i)$ before the command,
    then $\nu(Q_i) = (p+1) \nu(x_i)$ after it.
    Intuitively, if the simulated instruction was $p$, then we move to $p + 1$. That is,
    unless $s_i$ is non-zero. Otherwise, we move to the instruction dictated by the
    test-and-jump instruction that was just simulated. The second and third commands take into account the update of the value of counter $x_i$ into our state-encoding too.
    
    It is easy to see that, for all $N \in \mathbb{N}$, $S$ satisfies
	\[
		\llbracket S\rrbracket^N (\nu)(Q_i)
    	\geq 
    	m \cdot \llbracket S\rrbracket^N (\nu)(x_i)
	\]
    for any $1 \leq i \leq n$ if and only if $\calC$ halts in at most $N$ steps. To conclude the proof it suffices to observe
    that we can add a variable that is assigned $m$-times the
    value of $x_i$ at the end of the SLP to obtain an
    instance of \textsc{SLPPow} whose answer is positive
    if and only if the above inequality holds.
    \qed